\pdfoutput=1
\documentclass[aps,pra,superscriptaddress,amsmath,amssymb,twocolumn,footinbib]{revtex4-1}

\usepackage[T1]{fontenc}
\usepackage[utf8]{inputenc}

\usepackage{graphicx}
\usepackage[dvipsnames,svgnames]{xcolor}
\usepackage{braket}

\usepackage{newpxtext,newpxmath}

\usepackage{bbm,amsthm,amsmath}
\usepackage{enumerate}
\usepackage{booktabs,multirow}
\usepackage{mathtools,mathrsfs}
\usepackage[percent]{overpic}
\usepackage{microtype,hyphenat}
\usepackage{array,adjustbox,afterpage,breakurl}

\newtheoremstyle{red}{}{}{\itshape}{}{\color{red!80!black}\bfseries}{.}{ }{}

\definecolor{darkred}{rgb}{0.57,0,0.12}
\usepackage[colorlinks=true, linkcolor=MidnightBlue, urlcolor=darkred!75!black, citecolor=MidnightBlue, breaklinks=true]{hyperref}

\newcommand{\nc}{\newcommand}
\usepackage{cleveref}
\usepackage[most,breakable]{tcolorbox}

\AtBeginDocument{
\heavyrulewidth=.08em
\lightrulewidth=.05em
\cmidrulewidth=.03em
\belowrulesep=.65ex
\belowbottomsep=0pt
\aboverulesep=.4ex
\abovetopsep=0pt
\cmidrulesep=\doublerulesep
\cmidrulekern=.5em
\defaultaddspace=.5em
}

\usepackage{etoolbox}
\hypersetup{linktoc=all}

\nc{\ketbra}[2]{\ket{#1}\bra{#2}}

\DeclareMathOperator{\Tr}{Tr}

\DeclareMathOperator{\supp}{supp}

\DeclareMathOperator{\conv}{conv}

\DeclareMathOperator{\sspan}{span}

\newcommand{\proj}[1]{\ket{#1}\!\bra{#1}}

\newcommand{\lnorm}[2]{\left\lVert#1\right\rVert_{\ell_{#2}}}

\nc{\TT}[1]{T^{({#1})}_\I}
\nc{\TTJ}[1]{T^{({#1})}_\J}

\newcommand{\C}{\mathcal{C}}
\newcommand{\U}{\mathcal{U}}

\newcommand{\V}{\mathcal{V}}

\newcommand{\N}{\mathcal{N}}
\newcommand{\I}{\mathcal{I}}
\newcommand{\Q}{\mathcal{Q}}

\newcommand{\<}{\left\langle}
\renewcommand{\>}{\right\rangle}

\renewcommand{\bar}{\;\rule{0pt}{9.5pt}\right|\;}

\newcommand{\lset}{\left\{\left.}
\newcommand{\rset}{\right\}}

\DeclareMathOperator*{\argmin}{arg\,min}

\newcommand{\RR}{\mathbb{R}}
\newcommand{\CC}{\mathbb{C}}

\newcommand{\HH}{\mathbb{H}}
\newcommand{\DD}{\mathbb{D}}

\newcommand{\ve}{\varepsilon}

\newcommand{\cbraket}[1]{\left|\braket{#1}\right|}

\newcommand{\id}{\mathbbm{1}}

\renewcommand{\O}{\mathcal{O}}

\newcommand{\J}{\mathcal{J}}

\nc{\PPTPPR}{\text{\rm PPTP}_+}
\nc{\PPTPR}{\text{\rm PPTP}}

\nc{\ppt}{\text{\rm\sffamily PPT}}
\nc{\pptp}{\text{\rm\sffamily PPT}_{+}}

\newcommand{\PPT}{\ppt}

\nc{\PPTRP}{{{\PPT^{\hspace{0.1em}\prime}_+}}}
\nc{\PPTRPPR}{{\text{\rm PPTP}^{\hspace{0.1em}\prime}_+}}
\newcommand{\SEP}{\text{\rm\sffamily SEP}}

\nc{\wt}{\widetilde}

\newcommand{\F}{\mathcal{F}}

\newcommand*{\cL}{\mathcal{L}}

\def\gbm#1{{\let\phi\upphi \let\lambda\uplambda \let\mu\upmu \let\rho\uprho \let\sigma\upsigma \let\tau\uptau \let\theta\uptheta \let\eta\upeta \let\kappa\upkappa \bm{#1}}}

\newcommand{\be}{\begin{equation}}
\newcommand{\ee}{\end{equation}}

\nc{\logfloor}[1]{\left\lfloor {#1} \right\rfloor_{\log}}
\nc{\logceil}[1]{\left\lceil {#1} \right\rceil_{\log}}

\nc{\ins}{\mathrm{in}}
\nc{\outs}{\mathrm{out}}

\newtheorem{theorem}{Theorem}
\newtheorem{proposition}[theorem]{Proposition}
\newtheorem{corollary}[theorem]{Corollary}

\newtheorem{lemma}[theorem]{Lemma}

\theoremstyle{red}

\theoremstyle{definition}
\newtheorem*{remark}{Remark}

\let\oldproofname\proofname
\renewcommand{\proofname}{\rm\bf{\oldproofname}}

\let\nc\newcommand
  \nc{\MIO}{{\text{\rm MIO}}}
\nc{\DIO}{{\text{\rm DIO}}}
\nc{\SIO}{{\text{\rm SIO}}}
\nc{\IO}{{\text{\rm IO}}}
\nc{\lsetr}{\left\{\,}
\nc{\rsetr}{\right.\right\}}
\nc{\barr}{\,\rule{0pt}{9.5pt}\left|\;}

\nc{\FO}{\O}
\nc{\FOt}{\wt{\mathscr{F}}}

\renewcommand{\b}{{\bullet}}
\nc{\bb}{{\bullet\bullet}}
\DeclareMathOperator{\aff}{aff}

\nc{\rhom}{{\rho_{\text{\rm gold}}}}
\nc{\phim}{{\phi_{\max}}}
\nc{\phig}{{\phi_{\text{\rm gold}}}}
\nc{\phimp}{{\phi'_{\max}}}

 \nc{\Rg}{R_{\max}}
 \nc{\Rs}{R_{s}}
 \nc{\Ff}{f_\F}
 \nc{\f}{f}

 \let\dm\proj

  \nc{\Dmin}{D_{\min}}
  \nc{\Dmax}{D_{\max}}
   \nc{\Rmin}{R_{\min}}
   \nc{\Rmax}{R_{\max}}
   \nc{\Rmaxk}{\Rmax^{(k)}}
   \nc{\Rmink}{\Rmin^{(k)}}
   \nc{\Rminks}{\Rmin^{(k^\star)}}
   \nc{\Rmaxks}{\Rmax^{(k^\star)}}
   \nc{\Rsk}{\Rs^{(k)}}
   \nc{\Xik}{\Xi^{(k)}}

\nc{\rhodio}{{\rho\text{\rm-DIO}}}
\nc{\psidio}{{\psi\text{\rm-DIO}}}
\nc{\psimdio}{{\Psi_m\text{\rm-DIO}}}

\begin{document}

\title{Benchmarking one-shot distillation in general quantum resource theories}

\author{Bartosz Regula}
\email{bartosz.regula@gmail.com}
\affiliation{School of Physical and Mathematical Sciences, Nanyang Technological University, 637371, Singapore}
\author{Kaifeng Bu}
\email{kfbu@fas.harvard.edu}
\affiliation{Department of Physics, Harvard University, Cambridge, Massachusetts 02138, USA}
\author{Ryuji Takagi}
\email{rtakagi@mit.edu}
\affiliation{Center for Theoretical Physics and Department of Physics, Massachusetts Institute of Technology, Cambridge, Massachusetts 02139, USA}
\author{Zi-Wen Liu}
\email{zliu1@perimeterinstitute.ca}
\affiliation{Perimeter Institute for Theoretical Physics, Waterloo, Ontario N2L 2Y5, Canada}

\begin{abstract}%
We study the one-shot distillation of general quantum resources, providing a unified quantitative description of the maximal fidelity achievable in this task, and revealing similarities shared by broad classes of resources. We establish fundamental quantitative and qualitative limitations on resource distillation applicable to all convex resource theories.
We show that every convex quantum resource theory admits a meaningful notion of a pure maximally resourceful state which maximizes several monotones of operational relevance and finds use in distillation. 
We endow the generalized robustness measure with an operational meaning as an exact quantifier of performance in distilling such maximal states in many classes of resources including bi- and multipartite entanglement, multi-level coherence, as well as the whole family of affine resource theories, which encompasses important examples such as asymmetry, coherence, and thermodynamics.
\end{abstract}

\maketitle


\section{Introduction}

Effective exploitation of the advantages provided by quantum phenomena is a central aim of quantum information processing. The necessity to use many different properties of physical systems as the resources fueling such advantages --- examples being entanglement~\cite{horodecki_2009}, coherence~\cite{aberg_2006,baumgratz_2014,streltsov_2017}, asymmetry~\cite{gour_2008}, quantum thermodynamics (athermality)~\cite{brandao_2013}, purity~\cite{gour_2015}, steering~\cite{gallego_2015}, or magic~\cite{veitch_2014,howard_2017} --- motivated the development of formal theoretical frameworks which can characterize the various phenomena, dubbed \emph{resource theories}~\cite{horodecki_2012,coecke_2016,chitambar_2019}. Using this formalism, it was realized that many important features of quantum resources are in fact very general and can be characterized in a unified manner~\cite{brandao_2015,gour_2017,liu_2017,regula_2018,anshu_2018-1,takagi_2019-2,takagi_2019,uola_2019-1,liu_2019,vijayan_2019,takagi_2019-3}, which not only eliminates the need for resource-specific approaches, but often also simplifies and extends the insight one can obtain on a per-resource basis. Understanding which properties of resource theories are indeed universal, and which properties require more tailored methods can help streamline the description of the variety of phenomena of relevance to quantum information.

One of the most fundamental operational problems within a resource theory is to understand how the resource can be manipulated using free operations, that is, the quantum channels which are freely available in the physical setting of a given resource. Among these tasks, the problem of \emph{distillation} stands out as one of the most important in the practical exploitation of resources. The task is concerned with transforming arbitrary states into the form of pure, maximally resourceful target states which can act as a ``currency'' in various protocols. First considered in the resource theory of entanglement~\cite{bennett_1996,bennett_1996-3}, distillation has since become a cornerstone of other resource theories such as magic-state quantum computation~\cite{bravyi_2005} and coherence~\cite{winter_2016,regula_2017}. The problem with understanding distillation in general resource-theoretic settings is that a priori it is not clear whether a meaningful choice of a target state can always be defined, nor do we have general methods that can be applied to characterize this task in all resource theories. The recent work of Ref.~\cite{liu_2019} made progress in the description of distillation in a unified formalism, but the bounds obtained therein are not achievable in several important cases of interest --- such as the resource theories of asymmetry, coherence, or thermodynamics --- necessitating the development of a more general approach in order to encompass wider classes of resources. Additionally, previous results were concerned with the maximal rates of distillation rather than the precise characterization of the achievable performance; however, in the practically relevant one-shot setting where imperfect state transformations often incur large errors, it is crucial to understand how closely one can approximate a target state even when one cannot distill it perfectly.

In this work, we address the question of how to precisely benchmark the performance of a state in one-shot resource distillation by characterizing the maximal achievable fidelity in this task. We establish a comprehensive set of tools applicable to broad classes of finite-dimensional resource theories, not only recovering results known for many resources of interest such as entanglement, asymmetry, and thermodynamics in a unified fashion, but also extending them to hitherto uncharacterized resources. We reveal quantitative and qualitative similarities shared between many resource theories, relying only on their underlying convex structure. Our results establish universal limitations on the distillation of resources in practical settings by showing that the perfect distillation of a full-rank state is impossible in any resource theory. Furthermore, we show that every convex resource theory admits a maximally resourceful state relevant to distillation, and characterize its distillability exactly in wide families of resources. Notably, we connect the distillation of such maximal states with the generalized robustness measure, endowing this important resource quantifier with another fundamental operational interpretation. Finally, we extend the characterization of one-shot rates of distillation in Ref.~\cite{liu_2019} to include also resource theories which have not been considered in that setting.


\section{Resources and their measures}\label{sec:prelim}

We will use $\HH$ to denote the space of self-adjoint operators acting on a $d$-dimensional Hilbert space, and $\DD$ the subset of valid density operators. $\<A,B\> = \Tr(A^\dagger B)$ will denote the Hilbert-Schmidt inner product. For a pure state $\ket\psi$, we often use the shorthand $\psi \coloneqq \proj\psi$. 

A resource theory is typically defined as consisting of a set of free states $\F$ and a set of free operations $\O$. In this work, we will be concerned with the largest meaningful class of free operations --- resource non-generating maps --- defined to be all channels $\Lambda$ which preserve the set of free states, in the sense that $\Lambda(\sigma) \in \F$ for any $\sigma \in \F$. We will simply use $\O$ to refer to the set of such channels. By considering the largest choice of free maps, our results therefore establish the ultimate limits on the achievable performance of any class of free operations. As for the set $\F$, we will make two intuitive assumptions: that the set $\F$ is convex, i.e. one cannot generate any resource by simply probabilistically mixing free states, and that it is closed, i.e. the limit of a sequence of free states must remain free. We refer to any resource theory satisfying the above assumptions as a convex resource theory. Note that one can certainly define a physical resource which does not admit a convex set of free states, but the assumption of convexity is natural in most settings and indeed the majority of physical theories of interest satisfy it~\cite{brandao_2015}.
When discussing the action of a quantum channel $\Lambda: \HH_\ins \to \HH_\outs$, we will implicitly assume that the set of free states $\F$ is defined in both the input and the output space, and we will not distinguish between them when no ambiguity arises.

Many possible choices of quantifiers can be defined within a general resource theory~\cite{regula_2018,chitambar_2019}. Out of these, we will make use of two measures based on entropic quantities~\cite{datta_2009}: the max-relative entropy $D_{\max} (\rho \| \sigma) = \log \inf \lset \lambda \bar \rho \leq \lambda \sigma \rset$, and the min-relative entropy $D_{\min} (\rho \| \sigma) = - \log \< \Pi_\rho, \sigma \>$ where $\Pi_\rho$ denotes the projection onto the support of $\rho$. One then often defines resource monotones by minimizing the two entropies over the set of free states~\cite{datta_2009,datta_2009-2}. However, here we will mostly not be concerned with the rates of transformations, but rather the geometric features of the set $\F$, which motivates us to define also the non-logarithmic versions of such measures:
\begin{equation}\begin{aligned}
  &D_{\max} (\rho) \coloneqq \inf_{\sigma \in \F} D_{\max} (\rho \| \sigma), &\Rmax (\rho) \coloneqq 2^{\Dmax(\rho)},\\
  &\Dmin (\rho) \coloneqq \inf_{\sigma \in \F} D_{\min} (\rho \| \sigma), &\Rmin (\rho) \coloneqq 2^{\Dmin(\rho)}.
\end{aligned}\end{equation}
Note that $\Rmax(\rho)$ and $\Rmin(\rho)$ will have a finite value for any $\rho$ if and only if $\F$ contains a full-rank state, in which case the infimum becomes a minimum. It is useful to notice that, for a pure state $\phi$, the quantity $\Rmin$ simplifies to $\Rmin^{-1}(\phi) = \sup_{\sigma \in \F} \< \phi, \sigma \>$ which is known as the support function of the set $\F$~\cite{rockafellar_1970}.
The quantity $\Rmax$ (or, specifically, $\Rmax(\cdot)-1$) is often referred to as the generalized robustness of a resource~\cite{vidal_1999}, and can alternatively be expressed as
\begin{equation}\begin{aligned}\label{eq:rob_primal}
  \Rmax(\rho) &= \inf \lsetr \lambda \geq 1\barr \rho + (\lambda-1) \omega \in \lambda\, \F,\; \omega \in \DD \rsetr
\end{aligned}\end{equation}
where $\lambda \F = \lset \lambda \sigma \bar \sigma \in \F \rset$. We will also make heavy use of the dual form of the generalized robustness (see e.g.~\cite{brandao_2005,regula_2018}),
\begin{equation}\begin{aligned}\label{eq:rob_dual}
  \Rmax (\rho) = \sup \lset \< \rho, W \> \bar W \geq 0,\; W \in \F^\circ \rset,
\end{aligned}\end{equation}
where we use the notation $\F^\circ \coloneqq \lset X \bar \< X, \sigma \> \leq 1 \; \forall \sigma \in \F \rset$ to denote the so-called polar set of $\F$. By constraining the state $\omega$ in the definition of the robustness~\eqref{eq:rob_primal} to be a free state, one can also define an alternative measure of the resourcefulness of a state called the standard (free) robustness~\cite{vidal_1999}:
\begin{equation}\begin{aligned}
  \Rs(\rho) \coloneqq& \inf \lsetr \lambda\geq 1 \barr \rho + (\lambda-1) \sigma \in \lambda\, \F,\; \sigma \in \F \rsetr\\
  =& \sup \lset \< \rho, W \> \bar W \in \F^\circ,\; \id - W \in \F^\circ \rset.
\end{aligned}\end{equation}

Although our results will apply to general convex resource theories, to establish tighter bounds we will often employ tailored methods which depend on the structure of the set $\F$. We will delineate in particular two types of resources:
\begin{enumerate}[(i)]
\item \textit{full-dimensional resource theories}, i.e. ones that satisfy $\sspan(\F) = \HH$, e.g. bi- and multipartite entanglement~\cite{zyczkowski_1998,horodecki_2009}, magic~\cite{bravyi_2005,veitch_2014}, multi-level coherence~\cite{ringbauer_2018};
\item \textit{reduced-dimensional resource theories}, i.e. ones for which $\sspan(\F) \subsetneq \HH$, e.g. asymmetry~\cite{gour_2008,marvian_2016}, coherence~\cite{baumgratz_2014}, thermodynamics~\cite{brandao_2013}, purity~\cite{gour_2015}.
\end{enumerate}
Full-dimensional resources theories can be understood as those which have a non-empty interior relative to the set of all density matrices, while any reduced-dimensional theory forms a zero-measure subset of quantum states. The geometric properties are visualized in Fig.~\ref{fig:resources}.
The crucial difference between the two cases is that any full-dimensional resource theory satisfies $\Rs(\rho) < \infty$ for all $\rho$, which is not the case in any reduced-dimensional theory (see Fig.~\ref{fig:resources}).

A particular case of reduced-dimensional resource theories are \textit{affine resource theories}, considered first in~\cite{gour_2017}. Recall that, given a set $\C$, its affine hull $\aff(\C)$ is the smallest affine subspace which contains $\C$, alternatively defined as the set composed of all combinations $\sum_i x_i C_i$ where $C_i \in \C$, $x_i \in \RR$ and $\sum_i x_i = 1$. A resource theory is then affine if $\F = \aff(\F) \cap \DD$. Indeed, most reduced-dimensional resource theories of interest are in fact affine: these include for instance asymmetry, coherence, and thermodynamics.

\begin{figure}[t]

    \includegraphics[width=7.2cm]{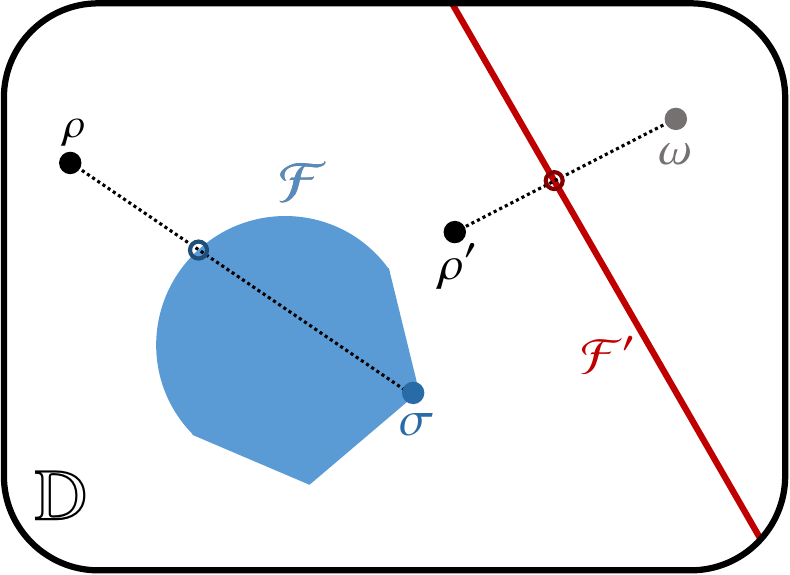}
 
    \caption{Schematic comparison between a full-dimensional resource theory $\F$ (in blue) and an affine reduced-dimensional resource theory $\F'$ (in red) within the set of all density matrices $\DD$ (within the black border). One can explicitly see that, for any state $\rho$, it is always possible to find a state $\sigma \in \F$ such that the convex combination of the two belongs to $\F$. On the other hand, the structure of the set $\F'$ --- which, in this case, is simply a line segment --- necessitates mixing $\rho'$ with another state $\omega \notin \F'$ in order for their convex combination to belong to $\F'$, as it is not possible to take $\omega \in \F$.}
    \label{fig:resources}
\end{figure}

A self-contained discussion of our methods and results is presented below, with some auxiliary proofs and additional extensions of the results deferred to the Appendix.


\section{Fidelity of general transformations}\label{sec:fidelity}

Resource distillation is a special case of a more general task, the purification of a resource state, in the sense that one is concerned with transforming an input state $\rho$ to some chosen resourceful pure state $\phi$.
To assess the performance of the resource non-generating operations $\O$ in such one-shot transformations, we will be interested in the maximal achievable fidelity with the target state $\phi$, which for simplicity we will always refer to as the \textit{fidelity of distillation}:
\begin{equation}\begin{aligned}
  F_\O(\rho \to \phi) \coloneqq \max_{\Lambda \in \O} \< \Lambda(\rho), \phi \>. 
\end{aligned}\end{equation}
Here, we limit ourselves to maps whose output dimension is equal to the dimension of $\phi$ in order for the inner product to be well defined. To characterize this quantity, we introduce the following optimization problem:
\begin{equation}\begin{aligned}
  G (\rho; k) \coloneqq \sup \lsetr \< \rho, W \> \barr 0 \leq W \leq \id,\; W \in \frac{1}{k} \F^\circ \rsetr
\end{aligned}\end{equation}
where $k \in [1,\infty)$, and we recall that $\F^\circ = \lset X \bar \< X, \sigma \> \leq 1 \; \forall \sigma \in \F \rset$. The quantity $G(\rho; k)$ can be seen to be measurable, in the sense that it can be obtained by measuring a suitably chosen observable $W$, and conversely any observable compatible with the given constraints can be used to bound the value of $G$. This can be compared with the dual form of the generalized robustness in Eq.~\eqref{eq:rob_dual}, and indeed we will later reveal some similarities between the quantities. We now proceed to establish general bounds for the fidelity of distillation using the quantity $G$.
\begin{theorem}\label{thm:G_bound}
In any convex resource theory, it holds that
\begin{equation}\begin{aligned}\label{eq:G_bound}
 F_\O(\rho \to \phi) \leq G (\rho; \Rmin(\phi) ).
\end{aligned}\end{equation}
If further it holds that $\Rs(\phi) < \infty$, then we have that
\begin{equation}\begin{aligned}
  F_\FO(\rho \to \phi) \geq G \left(\rho; \Rs(\phi) \right).
\end{aligned}\end{equation}
\end{theorem}
\begin{proof}
For any resource non-generating $\Lambda$, we have $X \in \F^\circ \Rightarrow \Lambda^\dagger(X) \in \F^\circ$. Noting that $\phi \in \Rmin^{-1}(\phi) \F^\circ$ by definition, combined with the fact that each $\Lambda$ is CPTP, we have that $\Lambda^\dagger (\phi) \in \lset W \bar 0 \leq W \leq \id,\; W \in  \Rmin^{-1}(\phi) \F^\circ \rset$, which gives the upper bound in \eqref{eq:G_bound}.

For the lower bound, consider the map
\begin{equation}\begin{aligned}
  \Lambda(X) = \< X, W \> \phi + \< X, \id - W \> \delta
\end{aligned}\end{equation}
where $\delta \in \F$ is a state such that $\phi + (\Rs(\phi)-1) \delta \in \Rs(\phi)\, \F$. Clearly, $\Lambda$ is CPTP when $0 \leq W \leq \id$ and resource non-generating iff $\< \sigma, W \> \leq \Rs^{-1}(\phi) \; \forall \sigma \in \F$ by convexity of $\F$. Therefore, for any $W \in \lset W \bar 0 \leq W \leq \id,\; W \in \Rs^{-1}(\phi) \F^\circ \rset$, there exists a resource non-generating channel such that $W = \Lambda^\dagger(\phi)$. The lower bound follows.
\end{proof}

\noindent This result allows one to precisely understand the performance achievable in the distillation from a given state in terms of the quantity $G(\rho;k)$, which we will show to enjoy some desirable properties. We will later see that the bounds are particularly useful in the description of transformations into a class of maximally resourceful states, and in particular that the upper and lower bounds coincide in many relevant cases, establishing a precise quantitative characterization of the fidelity of general transformations.

The above can be compared with the bounds in~\cite{liu_2019}, which instead of characterizing the fidelity of distillation focus on the one-shot rates for this task --- in contrast, the quantification of fidelity allows us to provide precise information about the error incurred in distillation.

A remarkable immediate consequence of Thm.~\ref{thm:G_bound} is as follows.

\begin{proposition}\label{prop:nogo}
Let $\rho$ be any state such that the support of $\rho$ contains a state $\sigma \in \F$. Then, $F_\O(\rho \to \phi) < 1$ for any resourceful state $\phi$.

In particular, exact one-shot distillation of any resource is impossible from a full-rank state $\rho$.
\end{proposition}
\begin{proof}
Suppose, towards a contradiction, that $F_\O(\rho \to \phi) = 1$. By Thm.~\ref{thm:G_bound}, this then implies that $G(\rho; \Rmin(\phi)) = 1$ where $\Rmin(\phi) > 1$ since $\phi \notin \F$. But $\< \rho, W \> = 1$ under the constraint $0 \leq W \leq \id$ can only hold when $W$ has the form $W = \Pi_\rho + G$
for some $0 \leq G \leq \id$ such that $\supp(G) \subseteq \ker(\rho)$, with $\Pi_\rho$ denoting the projection onto $\supp(\rho)$. Now, since there exists a state $\sigma \in \supp(\rho) \cap \F$, this means that $\< \Pi_\rho, \sigma \> = 1$. But this then implies that there cannot exist $W = \Pi_\rho + G$ such that $\< W, \sigma \> \leq \Rmin^{-1}(\phi) < 1\; \forall \sigma \in \F$, which means that the feasible set for $G(\rho; \Rmin(\phi))$ is empty --- a contradiction.
\end{proof}
The above establishes fundamental limitations on exact state transformations, and in particular distillation, in all convex resource theories. The Proposition shows in particular that even a small admixture of full-rank noise will prevent a state from being distillable in a single-shot transformation, necessitating approximate distillation instead.
This extends similar observations in the resource theories of entanglement~\cite{kent_1998,regula_2019-2} and coherence~\cite{fang_2018}, and shows this property to follow only from the underlying convex structure of $\F$.

We can furthermore show that the evaluation of $G(\rho;k)$ can be simplified for a pure $\rho$ in resource theories in which the set of free states $\F$ is defined as the convex hull of free pure states --- including, for instance, the theories of entanglement, coherence, and magic. We discuss this in detail in Appendix~\ref{app:pure}.


\subsection{Reduced-dimensional resource theories}

One can notice that the lower bound in Thm.~\ref{thm:G_bound} relies on the fact that $\Rs(\phi) < \infty$ which, as we discussed before, is only satisfied for all states in full-dimensional resource theories. Indeed, a similar problem concerns the characterization of rates of distillation in Ref.~\cite{liu_2019}, where lower bounds are obtained only for full-dimensional resources. To remedy this problem, we will now establish a description of distillation applicable also to reduced-dimensional resources.

To generalize the above considerations to resource theories for which $\sspan(\F) \subsetneq \HH$, let us define the \textit{affine polar set} of $\F$ as
\begin{equation}\begin{aligned}
  \F^\bullet \coloneqq \lset W \in \HH \bar \< W, \sigma \> = 1 \; \forall \sigma \in \F \rset.
\end{aligned}\end{equation}
The reason why the set $\F^\b$ is relevant only in reduced-dimensional resource theories is simply because for any full-dimensional $\F$, the affine polar trivializes to $\F^\b = \{\id\}$. Indeed, the smaller the dimension of $\sspan(\F)$, the larger the dimension of $\F^\b$. We will later see additional motivations to consider the set $\F^\b$ in the context of reduced-dimensional and affine resource theories. We remark that our definition of $\F^\b$ is different from the affine dual set considered by Gour in \cite{gour_2017} as we allow for non-positive operators in $\F^\b$.

Using the affine polar set, we then define
\begin{equation}\begin{aligned}
  G^\b(\rho; k) \coloneqq \sup \lsetr \< \rho, W \> \barr 0 \leq W \leq \id,\; W \in \frac{1}{k} \F^\b \rsetr
\end{aligned}\end{equation}
and use it to obtain the following bounds, which provide a generalization of our methods to reduced-dimensional resources.

\begin{theorem}\label{thm:G_bound_aff}
For any state such that $\phi \in \Rmin^{-1}(\phi) \F^\b$, we have
\begin{equation}\begin{aligned}
 F_\FO(\rho \to \phi) \leq G^\b(\rho ; \Rmin(\phi)).
\end{aligned}\end{equation}
Furthermore, for any state such that $\Rg(\phi) < \infty$, we have
\begin{equation}\begin{aligned}
  F_\FO(\rho \to \phi) \geq G^\b \left( \rho; \Rg(\phi) \right).
\end{aligned}\end{equation}
\end{theorem}

\noindent The proof follows similarly to Thm.~\ref{thm:G_bound}. For completeness, we include a full derivation in Appendix~\ref{app:proofs}. 

Let us remark about the property that $\phi \in \Rmin^{-1}(\phi) \F^\b$, that is, the overlap $\< \phi, \sigma \>$ is constant for all free states $\sigma$. One can always write $\aff(\F) = \{X\} + \cL$ for some fixed operator $X$ and linear subspace $\cL$, and any operator $Z$ which belongs to the orthogonal complement of $\cL$ clearly satisfies $\< Z, \sigma \> = \< Z, X \>$ which is constant for any $\sigma \in \F$. Whether there exists a choice of a pure target state satisfying this condition is in general resource-dependent, and needs to be verified for the particular theory in consideration. We will return to this issue shortly.

Note, however, that the lower bound in Thm.~\ref{thm:G_bound_aff} does not require that $\phi \in \Rmin^{-1}(\phi) \F^\b$ and thus is valid in all resource theories, providing a non-trivial lower bound for all reduced-dimensional resources. Furthermore, we remark that the condition of Thm.~\ref{thm:G_bound} that $\Rs(\phi) < \infty$ can also be satisfied in reduced-dimensional resource theories for the special case of states such that $\phi \in \aff(\F) \setminus \F$.

We will now proceed to apply the bounds we established to the context of resource distillation.



\section{Distillation in general\\resource theories}\label{sec:distill}

To talk about distillation of a resource, it is crucial to identify a target state that one wishes to distill, which is frequently assumed to be a ``maximally resourceful'' state or copies thereof. However, it is not obvious a priori whether such states can be identified in general resource theories. The concept of a \textit{golden state} was considered in~\cite{liu_2019}, where it was shown that many resource theories admit a state which maximizes both $\Rmin$ and $\Rmax$, making it useful in the characterization of operational tasks. Here, we extend this result and show that such states in fact exist in every single convex resource theory, not depending on any other properties of the set $\F$.

\begin{theorem}\label{thm:golden_existence}
In any convex resource theory, a pure state $\phi$ maximizes $\Rmax$ among all states of the given dimension if and only if it also maximizes $\Rmin$. Furthermore, for any such state we have $\Rmin(\phi) = \Rmax(\phi)$.

By convexity of $\Rmax$, such a choice of pure $\phi$ always exists. We call any such state a \emph{golden state} and denote it by $\phig$.
\end{theorem}
\begin{proof}
($\Rightarrow$) We will first consider a general case of a (possibly mixed) state ${\wt{\rho}}$ maximizing the robustness. Let $M = \Rg({\wt{\rho}})$, and let $W \geq 0,\; W \in \F^\circ$ be an optimal witness for $\Rg({\wt{\rho}})$ satisfying $\Rg({\wt{\rho}}) = \< W, {\wt{\rho}} \>$. Since $M$ is the largest value that the robustness can take, we have that $\< W, \omega \> \leq M$ for all density matrices $\omega$ and so $W \leq M \id$. Now, since ${\wt{\rho}}$ maximizes the robustness, we have that every pure state $\ket\psi \in \supp({\wt{\rho}})$ is also a maximizer due to the convexity of $\Rg$, and so $\<W, \psi\> = M \; \forall \ket\psi \in \supp({\wt{\rho}})$. But the constraint $W \leq M \id$ then implies that $W\ket\psi = M \ket\psi$ for any such $\ket\psi$. Denoting by $\Pi_{\wt{\rho}}, \Pi_{\wt\rho}^\perp$ the projection operators onto $\supp({\wt{\rho}})$ and $\ker({\wt\rho})$, respectively, we then have
\begin{equation}\begin{aligned}
W = (\Pi_{\wt\rho} + \Pi_{\wt\rho}^\perp) W (\Pi_{\wt\rho} + \Pi_{\wt\rho}^\perp) = M \Pi_\rho + \Pi_{\wt\rho}^\perp W\Pi_{\wt\rho}^\perp
\end{aligned}\end{equation}
which in particular means that an optimal $W$ necessarily takes the form $  W = M \Pi_{\wt{\rho}} + G$
with $\supp(G) \subseteq \ker({\wt{\rho}})$, $G \geq 0$.

Noting that a pure maximizer of $\Rg$ always exists due to its convexity, let us now restrict ourselves to the case when ${\wt{\rho}} = \wt{\phi}$ is pure. We will use the following lemma.
\begin{lemma}[{\cite{cavalcanti_2005,regula_2018}}]\label{lem:RgFf_bound}For any state $\rho$, $\displaystyle \Rg(\rho) \geq \Tr(\rho^2) / \sup_{\sigma \in \F} \< \rho, \sigma \>$, 
and in particular $\Rmax(\psi) \geq \Rmin(\psi)$ for pure states.\end{lemma}
\noindent The proof is obtained simply by choosing $\rho/(\sup_{\sigma \in \F} \< \rho, \sigma \>)$ as a feasible solution for the dual form of the robustness~\eqref{eq:rob_dual}.

Using the Lemma, we have that $\Rmin(\wt{\phi}) \leq M$. On the other hand, we know that an optimal witness must have the form $W = M \wt{\phi} + G$ for some $G \geq 0$, which gives $\<\wt{\phi}, \sigma\> \leq \frac{1}{M} \< W, \sigma \> \leq \frac{1}{M}$ for any $\sigma \in \F$ since $W \in \F^\circ$. From this, we have that $\Rmin(\wt{\phi}) \geq M$ and so equality holds.

Now, let $\omega$ be a maximizer of $\Rmin$. By the joint convexity of $\Dmin(\cdot\|\cdot)$~\cite{datta_2009} and the fact that $2^x$ is a non-decreasing convex function, $\Rmin$ is convex~\cite{hiriart-urruty_1993} and so we can take $\omega = \psi$ to be pure. 
If it were the case that $\Rmin(\psi) > \Rmin(\wt{\phi})$, by Lemma~\ref{lem:RgFf_bound} we would have that
\begin{equation}\begin{aligned}
  \Rg(\psi) \geq \Rmin(\psi) > \Rmin(\wt{\phi}) = M,
\end{aligned}\end{equation}
which would contradict the assumption that $\wt{\phi}$ maximizes $\Rg$. Therefore, $\wt{\phi}$ maximizes $\Rmin$ as required.

($\Leftarrow$) A very similar argument holds for the converse statement. If we take $\wt{\phi}$ to be a maximizer of $\Rmin$ among all states, then it also has to maximize $\Rmax$; were it not the case, and some other $\psi$ maximized $\Rmax$ with $\Rmax(\psi) > \Rmax(\wt{\phi})$, then from the above reasoning we would have that
\begin{equation}\begin{aligned}
    \Rmin(\psi) = \Rmax(\psi) > \Rmax(\wt{\phi}) \geq \Rmin(\wt{\phi})
\end{aligned}\end{equation}
which would contradict the assumption that $\wt{\phi}$ maximizes $\Rmin$.
\end{proof}

In resources such as entanglement and coherence, such golden states are used precisely as the target states for distillation. To justify the choice of $\phig$ as a distillation target in more general cases, let us remark important general consequences of Theorem~\ref{thm:golden_existence}. Firstly, in many relevant resource theories, it holds that $\Rg(\phi) = \Rs(\phi)$ for any pure state $\phi$; this includes several fundamental examples such as the resource theory of bipartite entanglement~\cite{vidal_1999,steiner_2003}, bipartite entanglement of Schmidt rank $k$~\cite{johnston_2018}, the resource theory of non-positive partial transpose (see Appendix~\ref{app:ppt}), as well as multi-level quantum coherence~\cite{johnston_2018}. In any such theory, for all golden states we have $\Rmin(\phig) = \Rmax(\phig) = \Rs(\phig)$, and so the upper and lower bounds in Thm.~\ref{thm:G_bound} coincide. In fact, this property can be satisfied also in theories for which it is not necessarily the case that $\Rg(\phi) = \Rs(\phi)$ for all pure states --- in Appendix~\ref{app:multipartite}, we show that golden states in the resource theory of genuine multipartite entanglement satisfy $\Rmax(\phig) = \Rs(\phig)$, allowing us to immediately apply Thm.~\ref{thm:G_bound} in this theory as well.

In all of the above cases, our result establish an \emph{exact} characterization of the fidelity of distillation in terms of the quantity $G$.

\begin{corollary}\label{cor:G_exact}
In any resource theory such that $\Rg(\phig) = \Rs(\phig)$, it holds that
\begin{equation}\begin{aligned}
  F(\rho \to \phig) &= G(\rho; \Rmin(\phig))\\ &= G(\rho; \Rmax(\phig))
\end{aligned}\end{equation}
for any golden state $\phig$.
\end{corollary}

Similarly, a general property of this kind holds in reduced-dimensional resource theories in which the overlap $\< \phig, \sigma \>$ is constant for all free states $\sigma$.
\begin{corollary}\label{cor:G_exact_aff}
In any convex resource theory, it holds that
\begin{equation}\begin{aligned}
  F(\rho \to \phig) &= G^\b(\rho; \Rmin(\phig))\\ &= G^\b(\rho; \Rmax(\phig))
\end{aligned}\end{equation}
for any golden state $\phig$ such that $\phig \in \Rmin^{-1}(\phig) \F^\b$.
\end{corollary}
The condition that $\phig \in \Rmin^{-1}(\phig) \F^\b$ will, in general, be theory-dependent. Although golden states in many theories of interest such as coherence or thermodynamics do indeed satisfy the above assumption, one can also construct counterexamples, and indeed we explicitly demonstrate in Appendix~\ref{app:asym} that the golden states in the theory of asymmetry do not obey the condition.

We further remark that the condition $\Rmax(\phi) = \Rmin(\phi)$ can be satisfied by a pure state $\phi$ which is not a golden state, as evidenced e.g. by the so-called Clifford magic states within the resource theory of magic~\cite{bravyi_2019}, constituting a whole class of pure states for which the two quantifiers are equal. Provided that the conditions of Cor.~\ref{cor:G_exact} or Cor.~\ref{cor:G_exact_aff} are met within the given resource theory, the transformations of arbitrary states into such states can then be characterized exactly as the upper and lower bounds of  Thm.~\ref{thm:G_bound} or \ref{thm:G_bound_aff} coincide.

Note also that for any golden state we can quantify exactly the relative entropy $D(\phig) \coloneqq \inf_{\sigma \in \F} D(\phig \| \sigma)$ as $D(\phig) = \Dmax(\phig) = \Dmin(\phig)$ due to the fact that $D(\cdot\|\cdot)$ is upper and lower bounded by, respectively, $\Dmax$ and $\Dmin$~\cite{datta_2009}. In particular, any golden state is also a maximiser of the relative entropy measure.

Among all golden states, we will identify a family of states which admits a particularly simplified and insightful characterization --- the maximal golden states, i.e. golden states of the same dimension as the input state $\rho$, which are the most resourceful states that one could hope to distill from the given input.


\section{Maximal golden state distillation}\label{sec:maximal}

Throughout this section, we use $\HH_\ins$ to explicitly denote the input space of the distillation protocol (or a space isomorphic thereto), i.e. $\rho \in \HH_\ins$.

\begin{theorem}\label{thm:fid_max}
For any state $\phi$ in any convex resource theory,
\begin{equation}\begin{aligned}
  F_\O(\rho \to \phi) \leq \frac{\Rmax(\rho)}{\Rmin(\phi)}.
\end{aligned}\end{equation}
Furthermore, let $\phim \in \HH_\ins$ be a golden state in the input space. Provided that $\Rs(\phim) < \infty$, it holds that
\begin{equation}\begin{aligned}
  F_\O(\rho \to \phim) \geq \frac{\Rmax(\rho)}{\Rs(\phim)}.
\end{aligned}\end{equation}
The result also holds if $\phim$ is a state which maximizes $\Rs$ in the input space.
\end{theorem}
\begin{proof}
For any resource non-generating and CPTP map $\Lambda$, we have that $\Lambda^\dagger(\phi) \geq 0$ and $\Lambda^\dagger(\phi) \in \Rmin^{-1}(\phi) \F^\circ$, which means that $\Lambda^\dagger(\phi) \Rmin(\phi)$ is a feasible solution for the robustness $\Rg$, giving $\Rg(\rho) \geq \Rmin(\phi) \< \Lambda^\dagger(\phi), \rho \>$, from which the upper bound on $F_\O(\rho \to \phi)$ follows.

For the other inequality, take $\delta \in \F$ as an optimal state such that $\phim + (\Rs(\phim)-1) \delta \in \Rs(\phim) \F$, take $W \geq 0, W \in \F^\circ$ as an optimal operator such that $\Rg(\rho) = \< W, \rho \>$, and define the map
\begin{equation}\begin{aligned}
  \Lambda(X) = \< \frac{W}{\Rs(\phim)}, X \> \phim + \< \id - \frac{W}{\Rs(\phim)}, X \> \delta.
\end{aligned}\end{equation}
Using the fact that
\begin{equation}\begin{aligned}\label{eq:app_W_Rmax}
  \< \frac{W}{\Rs(\phim)}, \omega \> \leq \frac{\Rg(\omega)}{\Rs(\phim)} \leq \frac{\Rg(\omega)}{\Rg(\phim)} \leq 1
\end{aligned}\end{equation}
for any state $\omega$ as $\phim$ maximizes $\Rmax$ by assumption, we have that $\frac{W}{\Rs(\phim)} \leq \id$ and so $\Lambda$ defines a valid CPTP map. Moreover, we have for any $\sigma \in \F$ that $\<\frac{W}{\Rs(\phim)}, \sigma \> \leq \frac{1}{\Rs(\phim)}$, and so by definition of the robustness $\Rs$ and convexity of the set $\F$ it holds that $\Lambda(\sigma) \in \F$. $\Lambda$ is therefore a resource non-generating map, and we have
\begin{equation}\begin{aligned}
  F_\O(\rho \to \phim) \geq \< \Lambda(\rho), \phim \> = \frac{\Rg(\rho)}{\Rs(\phim)}.
\end{aligned}\end{equation}
From Eq.~\eqref{eq:app_W_Rmax} it can be noticed that the same result applies when $\phim$ is chosen as a state which maximizes $\Rs$.
\end{proof}

In the particular case of resource theories such that the two robustness measures coincide, the bounds give an exact characterization of the distillation fidelity.
\begin{corollary}
In any resource theory such that $\Rmax(\phi) = \Rs(\phi)$ for all pure states,
\begin{equation}\begin{aligned}
  F_\O(\rho \to \phim) = \frac{\Rmax(\rho)}{\Rmin(\phim)} = \frac{\Rmax(\rho)}{\Rmax(\phim)},
\end{aligned}\end{equation}
and so the robustness $\Rg(\rho)$ quantifies exactly the fidelity of distillation $F_\O(\rho \to \phim)$.
\end{corollary}
This recovers results shown specifically for the case of bipartite entanglement theory~\cite{regula_2019-2} and extends them to more general classes of resources. Our results in App.~\ref{app:multipartite} show that this property is satisfied also for genuine multipartite entanglement. Importantly, in any such theory it provides another operational meaning to the generalized robustness $\Rmax(\rho)$, establishing it as an exact quantifier of the fidelity achievable in distilling the maximal state $\phim$. This complements other operational applications of this measure in different settings~\cite{takagi_2019-2,liu_2019,takagi_2019,seddon_2020}.

Another interesting consequence of the Theorem is as follows.
\begin{corollary}\label{cor:interconvert}
Consider the case when $\Rmax(\phim) = \Rs(\phim)$ and take any state $\rho$ such that $\Rs(\rho) = \Rs(\phim)$. Then, there exists maps $\Lambda_1, \Lambda_2 \in \O$ such that $\Lambda_1(\rho) = \phim$ and $\Lambda_2(\phim) = \rho$.
\end{corollary}
\begin{proof}The existence of $\Lambda_1$ follows from the above Theorem, and the existence of $\Lambda_2$ from~\cite{liu_2019}.\end{proof}
The above can be understood as the fact that the state $\phim$ is unique up to an application of a free operation. In particular, we can talk about the distillation of $\phim$ without loss of generality: as long as one is concerned with the distillation of the \textit{maximal} resource state, any such state can be interconverted with $\phim$ for free. Although the conditions of Cor.~\ref{cor:interconvert} are satisfied in the resource theories that we mentioned previously --- bipartite and genuine multipartite entanglement, non-positive partial transpose, and multi-level coherence --- it is also known that there exist resource theories such that $\Rmax(\phim) \neq \Rs(\phim)$, notably  magic~\cite{regula_2018} and a weaker form of multipartite entanglement (non--full-separability)~\cite{contreras-tejada_2019}, in which cases such interconversion might not be possible.

\subsection{Reduced-dimensional and affine resources}

In order to characterize distillation in reduced-dimensional resource theories, we now introduce the \textit{affine robustness}, defined as the robustness with respect to the affine hull $\aff(\F)$, i.e.
\begin{equation}\begin{aligned}
  \Rg^\b (\rho) \coloneqq& \inf \lsetr \lambda \geq 1 \barr \rho + (\lambda-1) \omega \in \lambda \aff(\F),\; \omega \in \DD \rsetr.
\end{aligned}\end{equation}
Note that each $X \in \aff(\F)$ necessarily has $\Tr X = 1$.
To derive the dual form of this measure in analogy to the dual form of the generalized robustness, it will be important to note the following relation.
\begin{lemma}\label{lemma:affine_hull}
For any closed and convex set $\F$, $\aff(\F)^\circ = (-\infty,1]\, \F^\b$.
\end{lemma}
\begin{proof}
If $Z \in (-\infty,1]\, \F^\b$, that is, there exists some $\mu \in (-\infty,1]$ such that $\<Z, \sigma \> = \mu \; \forall \sigma \in \F$, then clearly $\< Z, X \> = \mu \leq 1$ for all $X \in \aff(\F)$. On the other hand, let $Z \in \aff(\F)^\circ$ and assume towards contradiction that $Z \notin \mu \F^\b$ for any $\mu \in (-\infty,1]$. This means that $Z$ does not have a constant value on the set $\F$, so there exists a choice of $\mu$ and $c \neq 0$ such that $\< Z, \sigma\> = \mu$ and $\<Z ,\sigma' \> = \mu - c$ for some $\sigma, \sigma' \in \F$. But then we can choose $x \in \RR$ such that $xc > (1-\mu+c)$ to obtain an affine combination with $\<Z, x \sigma + (1-x) \sigma' \> > 1$, which contradicts our assumption that $Z \in \aff(\F)^\circ$.
\end{proof}

Taking the Lagrange dual analogously to Eq.~\eqref{eq:rob_dual} and using Lemma~\ref{lemma:affine_hull}, we obtain
\begin{equation}\begin{aligned}
  \Rg^\b (\rho) = \sup \lset \< \rho, W \> \bar W \geq 0,\; W \in \F^\b \rset
\end{aligned}\end{equation}
by noting that an optimal $W \in (-\infty,1]\, \F^\b$ must satisfy $W \in \F^\b$. Furthermore, we observe the following.
\begin{lemma}
$\Rg^\b (\rho) = \Rg(\rho) \; \forall \rho \in \DD$ if and only if the resource theory is affine.
\end{lemma}
\begin{proof}
If the theory is affine, then $\F = \aff(\F) \cap \DD$ by definition, and equality between the measures follows immediately. If the theory is not affine, this means that there exists a state $\rho \in \aff(\F)$ such that $\rho \notin \F$, which implies that $\Rg^\b(\rho) = 1$ but $\Rg(\rho) > 1$.
\end{proof}
We then obtain the following general bound.

\begin{theorem}\label{thm:fid_max_aff}
For any state such that $\phi \in \Rmin^{-1}(\phi) \F^\b$, we have
\begin{equation}\begin{aligned}
  F_\O(\rho \to \phi) \leq \frac{\Rmax^\b(\rho)}{\Rmin(\phi)}.
\end{aligned}\end{equation}
For any state $\phim \in \HH_\ins$ which maximizes $\Rg$ with $\Rg(\phim) < \infty$, we have
\begin{equation}\begin{aligned}
 F_\O(\rho \to \phi_{\max}) \geq \frac{\Rmax^\b(\rho)}{\Rmax(\phim)}.
\end{aligned}\end{equation}
\end{theorem}

An important corollary of the upper bound in Thm.~\ref{thm:fid_max} combined with the lower bound in Thm.~\ref{thm:fid_max_aff} is the following:
\begin{corollary}
In any affine resource theory, we have
\begin{equation}\begin{aligned}
  F_\O(\rho \to \phi_{\max}) = \frac{\Rmax(\rho)}{\Rmin(\phim)} = \frac{\Rmax(\rho)}{\Rmax(\phim)}
\end{aligned}\end{equation}
and so the robustness $\Rg(\rho) = \Rg^\b(\rho)$ quantifies exactly the fidelity of distillation $F_\O(\rho \to \phi_{\max})$.
\end{corollary}

We stress that this is a powerful result that applies immediately to broad classes of resources, including asymmetry, athermality, purity, and coherence, and generalizes previous results shown specifically for coherence~\cite{bu_2017,regula_2017}. The result requires no assumption about the maximal state $\phim$ beyond the fact that it maximizes the robustness $\Rmax$, and we know from Thm.~\ref{thm:golden_existence} that such a state always exists and is a meaningful choice of a maximally resourceful state. Therefore, the above fully characterizes the distillation of such maximal resources in every affine resource theory.

Analogously as before, we also have the following.
\begin{corollary}
Let $\F$ be an affine resource theory such that $\phim \in \HH_\ins$ satisfies $\phim \in \Rmin^{-1}(\phim) \F^\b$. For any $\rho$ such that $\Rg(\rho) = \Rg(\phim)$, there exist maps $\Lambda_1, \Lambda_2 \in \O$ such that $\Lambda_1(\rho) = \phim$ and $\Lambda_2(\phim) = \rho$.
\end{corollary}


\section{Distillation yield}\label{sec:yield}

Finally, let us consider how to characterize the yield of distillation, that is, a quantitative description of how much resource can be distilled up to an error $\ve$ in the distillation fidelity. This will allow us to explicitly recover the results of~\cite{liu_2019} using our bounds from Sec.~\ref{sec:fidelity}, and extend them to the previously unconsidered case of reduced-dimensional resources. To this end, we will connect the quantities $G$ and $G^\b$ with a fundamental operational quantity known as the hypothesis testing relative entropy~\cite{buscemi_2010,wang_2012}
\begin{equation}\begin{aligned}
D_H^\ve(\rho||\sigma) \coloneqq -\log\min \big\{  \< M, \sigma\> \,\big|\;& 0\leq M\leq \id, \\  & 1 - \< M, \rho\> \leq\ve \big\},
\end{aligned}\end{equation}
which can be understood as quantifying the error in distinguishing between $\rho$ and $\sigma$ in the problem of quantum hypothesis testing~\cite{wang_2012,tomamichel_2013,hayashi_2017}, or alternatively as a suitably smoothed variant of the min-relative entropy $D_{\min}(\rho\|\sigma)$~\cite{buscemi_2010}. In order to establish this connection, we will use the following characterization of the hypothesis testing relative entropy minimized over a set of operators.
\begin{lemma}[{\cite[Prop. 4]{regula_2019-2}}]\label{lemma:hyp_test_G}
For any closed and convex set of Hermitian operators $\Q$,
\begin{align}\label{eq:hyp_test_G}
  \inf_{X \in \Q} D^\ve_H(\rho \| X) = - \log \inf \big\{ \lambda \in \RR_+ \;\big|\;& \< \rho, W \> \geq 1-\ve,\nonumber\\\ & 0 \leq W \leq \id,\\& W \in \lambda \Q^\circ \big\}.\nonumber
\end{align}
\end{lemma}
Crucially, with a simple change of variables $\lambda \to \frac{1}{k}$, the constraints on the right-hand side of Eq.~\eqref{eq:hyp_test_G} can be identified with the quantity $G(\rho; k )$ that we introduced previously.
This immediately tells us in particular that
\begin{equation}\begin{aligned}
  \min_{\sigma \in \F} D^\ve_H(\rho \| \sigma) = \log \max \lset k \in \RR_+ \bar G(\rho; k) \geq 1 - \ve \rset.
\end{aligned}\end{equation}
Employing our bounds on the fidelity of distillation from Thm.~\ref{thm:G_bound}, we then recover exactly the result of Thm.~5 of~\cite{liu_2019} as follows.
\begin{corollary}
If $F_\O(\rho \to \phi) \geq 1 - \ve$, then
\begin{equation}\begin{aligned}
  \Dmin(\phi) \leq \min_{\sigma \in \F} D^\ve_H(\rho \| \sigma).
\end{aligned}\end{equation}
Conversely, if the resource theory is full-dimensional and 
\begin{equation}\begin{aligned}
  \log \Rs(\phi) \leq \min_{\sigma \in \F} D^\ve_H(\rho \| \sigma),
\end{aligned}\end{equation}
then $F_\O(\rho \to \phi) \geq 1 - \ve$.
\end{corollary}
\noindent The above connects the distillation of general resources with the quantity $D^\ve_H$.

To adapt this approach to reduced-dimensional theories, we will extend the definition of the hypothesis testing relative entropy $D_H^\ve(\rho\|\sigma)$, in that we will allow the operator $\sigma$ to be non-positive. In particular, recall from Lemma~\ref{lemma:affine_hull} that $\aff(\F)^\circ = \F^\b$, and consider Lemma~\ref{lemma:hyp_test_G} with the choice of $\Q = \aff(\F)$ to obtain
\begin{equation}\begin{aligned}
  \min_{X \in \aff(\F)} D^\ve_H(\rho \| X) = \log \max \lset k \in \RR_+ \bar G^\b(\rho; k) \geq 1 - \ve \rset.
\end{aligned}\end{equation}
We can then use our bounds in Thm.~\ref{thm:G_bound_aff} to get the following.
\begin{corollary}\label{cor:hyp_testing_affine}
Consider a state $\phi$ such that $\phi \in \Rmin^{-1}(\phi) \F^\b$. 
If $F_\O(\rho \to \phi) \geq 1 - \ve$, then
\begin{equation}\begin{aligned}
  \Dmin(\phi) \leq \min_{X \in \aff(\F)} D^\ve_H(\rho \| X).
\end{aligned}\end{equation}
Conversely, if
\begin{equation}\begin{aligned}
  \Dmax(\phi) \leq \min_{X \in \aff(\F)} D^\ve_H(\rho \| X),
\end{aligned}\end{equation}
then $F_\O(\rho \to \phi) \geq 1 - \ve$.
\end{corollary}
This extends the results of \cite{liu_2019} on the connection between the hypothesis testing relative entropy and resource distillation beyond full-dimensional resource theories. Curiously, one can then see that it becomes necessary to optimize $D^\ve_H$ over a set which goes beyond positive semidefinite operators. Although this makes it rather difficult to interpret $D^\ve_H$ in the context of hypothesis testing between two states, we have already seen such quantity applied successfully in several resource theories~\cite{fang_2019,regula_2017,regula_2019-2}, which motivates it as a meaningful notion worth further study.

We note that Cor.~\ref{cor:hyp_testing_affine} exactly recovers several fundamental cases of interest in which the upper and lower bounds coincide for golden states $\phig$, meaning that the maximal resource yield in the distillation of such states (as quantified by either $\Dmax$ or $\Dmin$) is precisely given by $\min_{X \in \aff(\F)} D^\ve_H(\rho\|X)$. This includes one-shot distillation of coherence~\cite{regula_2017} and one-shot work extraction in the resource theory of quantum thermodynamics~\cite{horodecki_2013,yungerhalpern_2016,liu_2019}. In particular, observe that in thermodynamics $\F = \{\tau\}$ is a single-element set consisting only of the Gibbs state $\tau$ and thus $\aff(\F) = \F$, meaning that the characterization in the Corollary recovers exactly the fundamental role of the hypothesis testing relative entropy $D^\ve_H(\rho \| \tau)$ in work extraction \cite{horodecki_2013,yungerhalpern_2016}.


\section{Discussion}

We have established a comprehensive characterization of the achievable fidelity of distillation in general convex quantum resource theories, showing it to be upper and lower bounded by a family of optimization problems $G(\rho;k)$, with the bounds becoming tight in the distillation of so-called golden states for many relevant quantum resources. In particular, we have shown every convex resource theory to admit a meaningful notion of such golden states. We then demonstrated that the distillation of maximally resourceful states is characterized by the generalized robustness measure in broad classes of resource theories including the whole family of affine resources, thus endowing this fundamental measure with a new and very general operational interpretation. The results additionally allowed us to provide the quantitative characterization of one-shot rates of distillation in reduced-dimensional resource theories, which has not been considered in previous work~\cite{liu_2019}.

Since distillation underlies the practical utilisation of many resources --- such as entanglement in communication scenarios \cite{horodecki_2009}, magic states in fault-tolerant quantum computation~\cite{bravyi_2005}, and coherence in information processing protocols~\cite{streltsov_2017} --- we expect our methods and results to immediately find use in benchmarking the performance of distillation in a wide variety of operational settings.

An interesting technical aspect of the work is that the quantum channels performing the transformation $\rho \to \phi$ in all of our proofs are so-called measure-and-prepare channels, which form a structurally simple, but operationally limited class of maps --- it is known that most state transformations are not achievable by measure-and-prepare channels only~\cite{buscemi_2019}. In all of the discussed cases where our bounds become tight, we can therefore reach an intriguing conclusion: the distillation $\rho \to \phi$ can be optimally performed with a (binary) measure-and-prepare transformation. Put in another way, the difference between the corresponding quantifiers $\Rmin(\phi)$ and $\Rmax(\phi)$ (or $\Rs(\phi)$) can be understood as quantifying how closely the given state $\phi$ can be reached by a free measure-and-prepare channel within the given resource theory, and indeed we have shown that the class of golden states in many theories can be optimally distilled with such measure-and-prepare maps. The wide applicability of our exact results shows that measure-and-prepare channels are a powerful tool for the distillation in general resource theories. 

Our work extends the insights provided by several recent works which described the common structure and features shared by general resource theories~\cite{brandao_2015,gour_2017,liu_2017,regula_2018,anshu_2018-1,takagi_2019-2,takagi_2019,uola_2019-1,liu_2019}, unifies the approaches to one-shot distillation which have been previously applied to investigate specific resources~\cite{brandao_2011,fang_2019,regula_2017,zhao_2019,wang_2020,regula_2019-2}, and contributes an important class of operational results which can be applied to wide families of relevant resource theories.


\begin{acknowledgments}
B.R. acknowledges the support of Nanyang Technological University as well as the National Research Foundation of Singapore Fellowship No. NRF-NRFF2016-02.
K.B. acknowledges the support of  the Templeton Religion Trust Grant TRT 0159 and ARO Grant W911NF1910302.
R.T. acknowledges the support of NSF, ARO, IARPA, and the Takenaka Scholarship Foundation.
Z.-W.L. acknowledges the support of Perimeter Institute for Theoretical Physics.
Research at Perimeter Institute is supported by the Government
of Canada through Industry Canada and by the Province of Ontario through the Ministry
of Research and Innovation.

\textit{Note. } --- During the completion of this work, Ref.~\cite{fang_2019-1} independently obtained results related to our Prop.~\ref{prop:nogo}.

\end{acknowledgments}

\appendix

\makeatletter
\makeatother

\renewcommand{\theequation}{\Alph{section}.\arabic{equation}}
\setcounter{equation}{0}
\setcounter{figure}{0}
\setcounter{table}{0}
\setcounter{section}{0}
\makeatletter


\section{Proofs of results in main text}\label{app:proofs}


\begingroup
\renewcommand{\thetheorem}{\ref{thm:G_bound_aff}}
\begin{theorem}
For any state such that $\phi \in \Rmin^{-1}(\phi) \F^\b$, we have
\begin{equation}\begin{aligned}\label{eq:G_bound_aff}
 F_\FO(\rho \to \phi) \leq G^\b(\rho ; \Rmin(\phi)).
\end{aligned}\end{equation}
Furthermore, for any state such that $\Rg(\phi) < \infty$, we have
\begin{equation}\begin{aligned}
  F_\FO(\rho \to \phi) \geq G^\b \left( \rho; \Rg(\phi) \right).
\end{aligned}\end{equation}
\end{theorem}
\endgroup
\begin{proof}
The proof proceeds analogously to Thm. \ref{thm:G_bound}. Since each $\Lambda \in \O$ is resource non-generating, we have $X \in \F^\b \Rightarrow \Lambda^\dagger(X) \in \F^\b$. Noting that $\phi \in \Rmin^{-1}(\phi) \F^\b$ by assumption, combined with the fact that $\Lambda$ is CPTP, we have that $\Lambda^\dagger (\phi) \in \lset W \bar 0 \leq W \leq \id,\; W \in \Rmin^{-1}(\phi) \F^\b \rset$, which gives the upper bound in \eqref{eq:G_bound_aff}.

For the lower bound, consider the map
\begin{equation}\begin{aligned}
  \Lambda(X) = \< X, W \> \phi + \< X, \id - W \> \delta
\end{aligned}\end{equation}
where $\delta$ is a state such that $\phi + (\Rg(\phi)-1) \delta = \Rg(\phi) \pi$ with $\pi \in \F$. Clearly, $\Lambda$ is CPTP as long as $0 \leq W \leq \id$; further, if $\< \sigma, W \> = \Rmax^{-1}(\phi) \; \forall \sigma \in \F$ then $\Lambda(\sigma) = \pi$ and hence $\Lambda$ is resource non-generating. Therefore, for any $W \in \lset W \bar 0 \leq W \leq \id,\; W \in \Rmax^{-1}(\phi) \F^\b \rset$, there exists a resource non-generating map such that $W = \Lambda^\dagger(\phi)$. The lower bound follows.
\end{proof}


\begin{remark}An alternative characterization of the result in Thm.~\ref{thm:golden_existence} which can be of independent interest is as follows: a pure state $\phi$ maximizes $\Rmax$ if and only if it minimizes the free fidelity $F_{\max}(\rho) \coloneqq \max_{\sigma \in \F} F(\rho, \sigma)$; this follows since $F_{\max}(\phi) = \Rmin(\phi)^{-1}$ and the minimizer of $F_{\max}$ can always be taken to be pure since, by the joint concavity of $\sqrt{F(\cdot,\cdot)}$~\cite{nielsen_2011}, $F_{\max}$ is a concave function~\cite{hiriart-urruty_1993}.\end{remark}

\begingroup
\renewcommand{\thetheorem}{\ref{thm:fid_max_aff}}
\begin{theorem}
For any state such that $\phi \in \Rmin^{-1}(\phi) \F^\b$, we have
\begin{equation}\begin{aligned}
  F_\O(\rho \to \phi) \leq \frac{\Rmax^\b(\rho)}{\Rmin(\phi)}.
\end{aligned}\end{equation}
For any state $\phim$ which maximizes $\Rg$ with $\Rg(\phim) < \infty$, we have
\begin{equation}\begin{aligned}
 F_\O(\rho \to \phi_{\max}) \geq \frac{\Rmax^\b(\rho)}{\Rmax(\phim)}.
\end{aligned}\end{equation}
\end{theorem}
\endgroup
\begin{proof}
By assumption, $\phi \in \Rmin^{-1}(\phi) \F^\b$. For any resource non-generating and CPTP map $\Lambda$, we then have that $\Lambda^\dagger(\phi) \geq 0$ and $\Lambda^\dagger(\phi) \in \Rmin^{-1}(\phi) \F^\b$, which means that $\Lambda^\dagger(\phi) \Rmin(\phi)$ is a feasible solution for the affine robustness $\Rg^\b$, giving $\Rg^\b(\rho) \geq \Rmin(\phi) \< \Lambda^\dagger(\phi), \rho \>$, from which the upper bound on $F_\O(\rho \to \phi)$ follows.

For the other inequality, take $\delta$ as an optimal state such that $\phim + (\Rg(\phim)-1) \delta \in \Rg(\phim) \F$, take $W \geq 0, W \in \F^\b$ as an optimal operator such that $\Rg^\b(\rho) = \< W, \rho \>$, and define the map
\begin{equation}\begin{aligned}
  \Lambda(X) = \< \frac{W}{\Rg(\phim)}, X \> \phim + \< \id - \frac{W}{\Rg(\phim)}, X \> \delta.
\end{aligned}\end{equation}
Using the fact that $\< \frac{W}{\Rg(\phim)}, \omega \> \leq \frac{\Rg(\omega)}{\Rg(\phim)} \leq 1$ for any state $\omega$ as $\phim$ maximizes $\Rg$ by assumption, we have that $\frac{W}{\Rg(\phim)} \leq \id$ and so $\Lambda$ defines a valid CPTP map. Moreover, we have for any $\sigma \in \F$ that $\<\frac{W}{\Rg(\phim)}, \sigma \> = \frac{1}{\Rg(\phim)}$, and so by definition of the robustness $\Rg$ it holds that $\Lambda(\sigma) \in \F$. $\Lambda$ is therefore a resource non-generating map, and we have
\begin{equation}\begin{aligned}
  F_\O(\rho \to \phim) \geq \< \Lambda(\rho), \phim \> = \frac{\Rg^\b(\rho) }{\Rg(\phim)}.
\end{aligned}\end{equation}
\end{proof}


\section{Simplification for pure states}\label{app:pure}

Consider the case when $\F = \conv \lset \proj\psi \bar \ket\psi \in \V \rset$ for some set $\V$ of free pure states. One can define the measure~\cite{regula_2018}
\begin{equation}\begin{aligned}
  \gamma(\ket\psi) &= \inf \lset \sum_i |c_i| \bar \ket\psi = \sum_i c_i \ket{v_i},\; \ket{v_i} \in \V \rset\\
  &= \sup \lset \cbraket{\psi|x} \bar \cbraket{x|v} \leq 1 \; \forall \ket{v} \in \V \rset,
\end{aligned}\end{equation}
which quantifies the resource content of any pure state --- e.g., in the resource theory of entanglement this is the sum of the Schmidt coefficients, and in the resource theory of magic this is the stabiliser extent. In particular, $\Rg(\psi) = \gamma(\ket\psi)^2$ for any pure state \cite{regula_2018}.

We then have the following.
\begin{proposition}
For any pure state $\psi$ and any $k$ it holds that
\begin{equation}\begin{aligned}
  &\sqrt{G(\psi ; k )} =\\
  &=\sup \lsetr \cbraket{\psi|\omega} \barr \cbraket{\omega|v} \leq \frac{1}{\sqrt{k}} \; \forall \ket{v} \in \V,\; \braket{\omega|\omega} \leq 1 \rsetr \\
  &= \min_{\ket\psi = \ket{x}+\ket{y}}  \frac{1}{\sqrt{k}}\, \gamma(\ket{x}) + \sqrt{\braket{y|y}}
\end{aligned}\end{equation}
where the optimization in the second line is over unnormalised vectors $\ket{x},\ket{y}$.

In particular, from Thm.~\ref{thm:G_bound}, if the resource theory is full-dimensional, we have for any $\phi$ that
\begin{equation}\begin{aligned}
  &\min_{\ket\psi = \ket{x}+ \ket{y}}  \frac{\gamma(\ket{x})}{\sqrt{\Rmin(\phi)}} + \sqrt{\braket{y|y}} \\
  \geq &\sqrt{F_\FO(\psi \to \phi)}\\
  \geq & \min_{\ket\psi = \ket{x}+ \ket{y}}  \frac{\gamma(\ket{x})}{\sqrt{\Rs(\phi)}} + \sqrt{\braket{y|y}}
\end{aligned}\end{equation}
and equality holds for golden states in any theory such that $\Rg(\psi) = \Rs(\psi)$ for any pure state.
\end{proposition}
\begin{proof}With a simple rearrangement, $G(\psi;k)$ can be written as
\begin{equation}\begin{aligned}
   k G (\psi; k) &= \sup \lsetr \< \psi, W \> \barr 0 \leq W \leq k \id,\; W \in \F^\circ \rsetr\\
   &= \sup \lsetr \< \psi, W \> \barr  W \geq 0,\;  W \in \left(\F \cup \frac{1}{k} \DD\right)^\circ \rsetr.
\end{aligned}\end{equation}
where we used that $(\mathcal{C} \cup \mathcal{D})^\circ = \mathcal{C}^\circ \cap \mathcal{D}^\circ$ for convex and closed sets. The set $\conv\left(\F \cup \frac{1}{k} \DD\right) \eqqcolon \Q_k$ can be noticed to be the convex hull of rank-one terms as $\Q_k = \conv \lset \proj{x} \bar \ket{x} \in \V \cup \N_k \rset$ where $\V$ is the set of free pure state vectors and $ \N_k \coloneqq \lsetr \ket{x} \barr \lnorm{\ket{x}}{2} = \frac{1}{\sqrt{k}} \rsetr$.

By Thm. 10 in~\cite{regula_2018}, for any pure state $\ket\psi$ we then have
\begin{equation}\begin{aligned}
  k G (\psi; k) &= \sup \lset \braket{\psi | W | \psi} \bar W \geq 0, W \in \Q_k^\circ \rset \\ &= \sup \lset \cbraket{\psi | w}^2 \bar \ket{w} \in (\V \cup \N_k)^\circ \rset\\
  &= \Gamma_{\V \cup \N_k}(\ket\psi)^2
\end{aligned}\end{equation}
where we use
\begin{equation}\begin{aligned}
  \Gamma_{\C}(\ket\psi) \coloneqq& \sup \lset \cbraket{\psi|x} \bar \ket{x} \in \C^\circ \rset\\
  =& \inf \lset \lambda \in \RR_+ \bar \ket\psi \in \lambda \conv(\C) \rset
\end{aligned}\end{equation}
to define the so-called convex gauge function of a set~\cite{rockafellar_1970}. Since $\V$ and $\N_k$ are both compact sets, by standard results in convex analysis (see e.g.~\cite[16.4.1 and 15.1.2]{rockafellar_1970}), this gauge can be obtained as
\begin{equation}\begin{aligned}
  \Gamma_{\V \cup \N_k}(\ket\psi) = \min_{\ket\psi = \ket{x} + \ket{y}} \Gamma_\V (\ket{x}) + \Gamma_{\N_k} (\ket{y}).
\end{aligned}\end{equation}
Using the fact that $\Gamma_\V$ is precisely $\gamma$~\cite{regula_2018} and $\Gamma_{\N_k} (\ket{y}) = \sqrt{k} \lnorm{\ket{y}}{2}$, the result follows. 
\end{proof}

This establishes a general operational application of the pure-state quantifier $\gamma$, as well as simplifies the computation of the achievable fidelity (since one only needs to optimize over vectors in the underlying Hilbert space instead of general Hermitian matrices). We note that $\gamma$ admits analytical or semi-analytical expressions in several resource theories such as bipartite entanglement~\cite{steiner_2003}, coherence~\cite{piani_2016}, as well as multi-level entanglement and coherence~\cite{regula_2018}, and indeed this can be used to obtain exact expressions for $G(\psi;k)$~\cite{regula_2017,regula_2019-2} in such cases.


\section{Golden states in the theory of asymmetry}\label{app:asym}

Suppose that a bipartite two-qubit system is equipped with the local Hamiltonian $H=\dm{1}$, which constitutes the global Hamiltonian $H_{12}=H\otimes\id + \id\otimes H = \dm{01}+\dm{10}+2\dm{11}$.
Let us consider the theory of asymmetry with $U(1)$ group defined by the unitary representation $\U_t:t\rightarrow e^{-iH_{12}t}\cdot e^{iH_{12}t}$.
The free states for this theory is the set of states that are invariant under such symmetric transformations:
\begin{eqnarray}
 \F&=&\lset \sigma\bar \sigma = \U_t(\sigma),\ \forall t\rset\\
 &=&\lset \sigma\bar [\sigma,H_{12}] = 0\rset\\
 &=&{\rm conv}\{\dm{00},\dm{11},\{\dm{\phi_{\alpha,\beta}}\}_{\alpha,\beta}\}
\end{eqnarray}
where $\{\dm{\phi_{\alpha,\beta}}\}$ denotes the set of pure states parameterized by $\alpha,\beta\in \CC$ as $\ket{\phi_{\alpha,\beta}}=\alpha\ket{01}+\beta\ket{10}$.

Let us write a general two qubit state as $\ket{\psi}=a\ket{00}+b\ket{01}+c\ket{10}+d\ket{11}$ and find the maximum fidelity with the free states: $F_{\max}(\ket{\psi})=\max_{\sigma\in\F}\bra{\psi}\sigma\ket{\psi}$, which gives $\Rmin(\psi) = 1/F_{\max}(\ket\psi)$.
Since the fidelity is linear in the free state and $\F$ is the convex hull of pure free states, it suffices to only consider pure free states.
The fidelity with the first two free pure states are straightforward. 
For the last case, note first that  
\begin{eqnarray}
|\braket{\phi_{\alpha,\beta}|\psi}|^2=|\alpha^*b+\beta^*c|^2\leq \left||\alpha||b|+|\beta||c|\right|^2
\label{eq:alpha beta}
\end{eqnarray} 
where the equality in the inequality is achieved by appropriately taking phases of $\alpha$ and $\beta$. 
This means that we can take $a,b,c,d,\alpha,\beta \geq 0$ without loss of generality, which we will assume from now on. 
In order to maximize \eqref{eq:alpha beta} with the constraints $\alpha^2+\beta^2=1$, consider the Lagrangian:
\begin{eqnarray}
 L = \alpha b + \beta c - \lambda (\alpha^2 + \beta^2 - 1),
\end{eqnarray}
which gives 
\begin{eqnarray}
 \frac{\partial L}{\partial \alpha} = b - 2\alpha\lambda,\ \frac{\partial L}{\partial \beta} = c - 2\beta\lambda,\ \frac{\partial L}{\partial \lambda} = \alpha^2+\beta^2 - 1.  
\end{eqnarray}
Setting all of them to be zero and manipulating the first two equations gives  
\begin{eqnarray}
 \tilde\alpha c - \tilde\beta b = 0\ \iff\ \tilde\beta = \frac{c}{b}\tilde\alpha.
\end{eqnarray}
where $\tilde\alpha,\ \tilde\beta$ are the solutions giving a local optimum. 
Combining this with the third equality, we get 
\begin{eqnarray}
 \tilde\alpha = \left[\left(\frac{c}{b}\right)^2+1\right]^{-\frac{1}{2}},\ \tilde\beta = \left[\left(\frac{b}{c}\right)^2+1\right]^{-\frac{1}{2}}, 
\end{eqnarray}
and 
\begin{eqnarray}
|\braket{\phi_{\tilde\alpha,\tilde\beta}|\psi}|^2 &=& \left(\frac{b}{\sqrt{\left(\frac{c}{b}\right)^2+1}}+\frac{c}{\sqrt{\left(\frac{b}{c}\right)^2+1}}\right)^2 \\
 &=& b^2+c^2
\end{eqnarray}
Hence, we obtain 
\begin{eqnarray}
 &&\max_{\alpha,\beta} |\braket{\phi_{\alpha,\beta}|\psi}|^2 =\\
 &=& \max\{|\braket{\phi_{0,1}|\psi}|^2,|\braket{\phi_{1,0}|\psi}|^2,|\braket{\phi_{\tilde\alpha,\tilde\beta}|\psi}|^2\}\quad\quad\\
 &=&\max\{c^2,b^2,b^2+c^2\} = b^2+c^2.
\end{eqnarray}

Thus, we eventually get:
\begin{eqnarray}
 F_{\max}(\ket{\psi})=\max\{a^2,d^2,b^2+c^2\}.
 \label{eq: max fidelity}
\end{eqnarray}

Our goal now is to find the min-fidelity state (golden state) such that $\ket{\phi_{\rm gold}}={\rm argmin}_{\ket{\psi}} F_{\max}(\ket{\psi})$, which is reduced to finding $a, b, c, d$ that minimizes \eqref{eq: max fidelity}. 
Keeping in mind the normalization condition $a^2+b^2+c^2+d^2=1$, first observation we can make is that $a=d$, because otherwise we can always reduce $\max\{a,d\}$ without changing $a^2+d^2$. 
Then, \eqref{eq: max fidelity} further reduces to

\begin{eqnarray}
 F_{\max}(\ket{\psi})&=&\max\{a^2,b^2+c^2\}\\
 &=&\max\{a^2,1-2a^2\}
 \label{eq: max fidelity2}
\end{eqnarray}
because of $d=a$ and the normalization condition. 
This is clearly minimized when 
\begin{eqnarray}
&&a^2=1-2a^2\\
\Rightarrow\  &&(a,b,c,d) = \left(\frac{1}{\sqrt{3}},\frac{\cos\theta}{\sqrt{3}},\frac{\sin\theta}{\sqrt{3}},\frac{1}{\sqrt{3}}\right)
\end{eqnarray}
where $\theta$ is an arbitrary angle, which gives a general golden state as
\begin{eqnarray}
\ket{\phi_{\rm gold}(\theta)}:=\frac{1}{\sqrt{3}}(\ket{00}+\cos\theta\ket{01}+\sin\theta\ket{10}+\ket{11}).
\end{eqnarray}

One can explicitly check that this state has less overlap with the set of free states than the intuitively ``most coherent" state $\ket{++}$ which has $\Rmin(\proj{++})=2<3=\Rmin(\phi_{\rm gold}(\theta))$.
It can be also immediately seen that these states do not have a constant overlap with all of the free states. We have thus shown the following.

\begin{proposition}
 Golden states for theory of asymmetry with U(1) group do not satisfy $\phig \in \Rmin(\phig)^{-1} \, \F^\b$ in general. 
\end{proposition}


\section{Golden states in the theory of non-positive partial transpose}\label{app:ppt}

For completeness, we provide proofs which we could not find explicitly derived in the literature, but which follow straightforwardly from known results.

\begin{lemma}
Let $\SEP$ denote separable states and $\PPT$ states with positive partial transpose in some bipartite Hilbert space. For any pure state $\psi$, we have
\begin{equation}\begin{aligned}
  \Rs^\SEP (\psi) = \Rs^\PPT (\psi) = \Rmax^\PPT (\psi).
\end{aligned}\end{equation}
\end{lemma}
\begin{proof}
On the one hand, we clearly have $\Rmax^\PPT \leq \Rs^\PPT \leq \Rs^\SEP$ for any state by the inclusion between the relevant sets. On the other hand, it is known that~\cite{vidal_1999}
\begin{equation}\begin{aligned}
  \Rs^\SEP (\psi) = \< \psi, d \Psi_d \>
\end{aligned}\end{equation}
where $\Psi_d$ is the maximally entangled state in the Schmidt basis of $\psi$. Clearly, $d \Psi_d$ is positive semidefinite, and for any PPT state $\sigma = \rho^{T_B}$ with $\rho$ a valid state we have
\begin{equation}\begin{aligned}
  \< \sigma, d \Psi_d \> = \< \rho^{T_B}, d \Psi_d \> = \< \rho, S \> \leq \lambda_{\max} (S) = 1
\end{aligned}\end{equation}
where we used that the partial transpose of $d \Psi_d$ is the swap operator $S$ which has eigenvalues $1$ and $-1$. This means that $d \Psi_d$ is a feasible solution for the dual form of $\Rmax^\PPT$, and so we have $\Rmax^\PPT (\psi) \geq \< \psi, d \Psi_d \>$ as desired.
\end{proof}

\begin{lemma}
For any maximally entangled state $\ket{\Psi_m} = \frac{1}{\sqrt{m}} \sum_i \ket{ii}$, we have
\begin{equation}\begin{aligned}
  \Rmin^\SEP (\Psi_m) = \Rmin^\PPT (\Psi_m) = \frac{1}{m}.
\end{aligned}\end{equation}
\end{lemma}
\begin{proof}
As above, for any $\sigma = \rho^{T_B}$ we have $\<\sigma, \Psi_m \> \leq \frac{1}{m}$, and on the other hand $\Rmin^\SEP (\Psi_m) = \frac{1}{m}$~\cite{shimony_1995}.
\end{proof}

The above results show in particular that $\Rs(\psi)=\Rmax(\psi)$ for any pure state in the resource theory of non-positive partial transpose and that the maximally entangled states form a family of valid golden states in this theory.


\section{Golden states in the theory of genuine multipartite entanglement}\label{app:multipartite}

An $n$-partite pure state is said to be biseparable if it is separable across any bipartition, i.e. if it can be written as $\ket\psi = \ket{\phi_1} \otimes \ket{\phi_2}$ where each $\ket{\phi_i}$ consists of at most $n-1$ parties. A state which is not biseparable is referred to as \textit{genuinely multipartite entangled}. Clearly, there are multiple inequivalent bipartitions along which the separability needs to be verified. We will use $\{B_k\}_{k=1}^{2^{n-1}-1}$ to denote the $k$th such bipartition, $\V_{k}$ to denote the set of all biseparable pure states in the $k$th bipartition $B_k$, and $\F_k = \conv \lset \proj{v} \bar \ket{v} \in \V_k \rset$ to denote the set of mixed biseparable states along the given bipartition. In other words, the set of all biseparable pure states is given by $\V \coloneqq \bigcup_{k} \V_k$ and biseparable mixed states as $\F = \conv ( \bigcup_{k} \F_k )$.

We will now show that the robustness measures $\Rmax$ and $\Rs$ are equal for the golden states $\phig$ (maximally entangled state) in this resource theory, which will establish genuine multipartite entanglement as another example where our results provide an exact characterization of one-shot distillation, as all bounds derived in our work are tight for $\phig$.

Below, we will use $\Rmax$ (analogously, $\Rs$ and $\Rmin$) to denote the corresponding measure with respect to the set of biseparable states $\F$, and $\Rmaxk$ (analogously, $\Rsk$ and $\Rmink$) to denote that the optimization is with respect to the set $\F_k$ of states separable along the $k$th bipartition. Recall that $\Rmaxk(\psi) = \Rsk(\psi)$ for any pure state, and the measures can be evaluated analytically in terms of the Schmidt coefficients of the state~\cite{vidal_1999,steiner_2003,harrow_2003}. We then have the following.

\begin{theorem}
Let $\phig$ be a golden state which maximizes $\Rminks$ for the choice of $k$ such that $\displaystyle k^\star = \argmin_k \Rmink(\phig)$. Then, $\phig$ also maximizes $\Rmin$, and
\begin{equation}\begin{aligned}
  \Rmax(\phig) = \Rs(\phig) &= \min_{k} \Rmaxk(\phig)\\ &= \min_{k} \Rsk(\phig).
\end{aligned}\end{equation}
\end{theorem}
\begin{proof}
Due to linearity of the inner product $\<\cdot,\cdot\>$, for any pure state we have that
\begin{equation}\begin{aligned}
  \Rmin^{-1}(\psi) &= \max_{\sigma \in \conv(\bigcup_{k} \F_k)} \< \psi, \sigma \> = \max_{\sigma \in \bigcup_{k} \F_k} \< \psi, \sigma \> \\&= \max_{k} \max_{\sigma \in \F_k} \< \psi, \sigma \> = \left( \min_k \Rmink(\psi) \right)^{-1}
\end{aligned}\end{equation}
which, for the choice of $\phig$, means that $\Rminks(\phig) = \Rmin(\phig)$. But then
\begin{equation*}\begin{aligned}
  \Rminks(\phig) &= \Rmin(\phig) \leq \Rmax(\phig) \leq \Rs(\phig) \\&\leq \min_{k} \Rsk(\phig) = \min_{k} \Rmaxk(\phig) \\&\leq \Rmaxks(\phig) = \Rminks(\phig) 
\end{aligned}\end{equation*}
and so all the quantities must be equal. Furthermore, it can be noticed that $\phig$ indeed maximizes $\Rmin$ --- were it not the case and there existed some other maximizer $\phi'$, we would have
\begin{equation}\begin{aligned}
  \Rmin(\phig) &< \Rmin(\phi') = \min_{k} \Rmink(\phi') \leq \Rminks(\phi') \\&\leq \Rminks(\phig) = \Rmin(\phig)
\end{aligned}\end{equation}
due to the fact that $\phig$ maximizes $\Rminks$ by assumption; this is a contradiction, so $\phig$ must maximize $\Rmin$.
\end{proof}

The above simplifies in particular when considering a system consisting of $n$ particles of the same dimension $d$, since the minimization of the golden states over bipartitions is no longer necessary. Specifically, using the fact that a maximally entangled state $\ket{\Psi^+}=\sum_i \frac{1}{\sqrt{d}} \ket{ii}$ is a golden state in any bipartition, the result shows that the generalized GHZ state $\ket{\mathrm{GHZ}} = \sum_{i=1}^{d}\frac{1}{\sqrt{d}} \ket{i}^{\otimes n}$ is a golden state and satisfies $\Rmin(\mathrm{GHZ}) = \Rmax(\mathrm{GHZ}) = \Rs(\mathrm{GHZ}) = d$, so the one-shot distillation of GHZ states can be characterized exactly with our results.

We remark that the same results hold for the theory of genuine multipartite non-positive partial transpose, that is, the natural extension of the theory of non-positive partial transpose to the case of genuine $n$-partite entanglement.

\bibliographystyle{apsrmp4-2}
\bibliography{main}

\end{document}